\documentclass[reqno,11pt]{article} 

\usepackage{amsmath,amsthm,amssymb,amscd,amstext,amsfonts}
\usepackage{mathtools}
\usepackage{mathrsfs}
\usepackage{thmtools}
\usepackage{latexsym}
\usepackage{verbatim}
\usepackage{framed}
\usepackage{graphicx}
\usepackage{stmaryrd}
\usepackage{enumerate}
\usepackage{fullpage}
\usepackage{bm}
\usepackage{color}
\usepackage{hyperref}
\usepackage{url}
\usepackage{physics}
\usepackage{multicol}
\usepackage{multirow}
\usepackage{tikz-cd}
\usetikzlibrary{decorations.pathmorphing}   
\usetikzlibrary{arrows}
\usepackage{makeidx}
\usepackage{makecell}		

\makeindex

\usepackage{caption} 
\usepackage[basic]{complexity}    

\newtheorem{theorem}{Theorem}

\newtheorem{lemma}[theorem]{Lemma}

\newtheorem{definition}[theorem]{Definition}

\DeclarePairedDelimiter\ceil{\lceil}{\rceil}   

\newcommand{\cF}{\mathcal F}
\newcommand{\cG}{\mathcal G}

\newcommand{\cI}{\mathcal I}

\newcommand{\cT}{\mathcal T}

\newcommand{\bbN}{\mathbb N}

\newcommand{\bbR}{\mathbb R}
\newcommand{\bbZ}{\mathbb Z}

\usepackage{nicematrix}   
\DeclareMathOperator{\OPT}{OPT}
\DeclareMathOperator{\ALG}{ALG}
\DeclareMathOperator{\FF}{FF}

\begin{document}

\title{Online Vector Bin Packing and Hypergraph Coloring Illuminated: Simpler Proofs and New Connections}

\author{Yaqiao Li\footnote{Concordia University, yaqiao.li@concordia.ca}, 
Denis Pankratov\footnote{Concordia University, denis.pankratov@concordia.ca}}

\maketitle

\begin{abstract}
   This paper studies the online vector bin packing (OVBP) problem and the related problem of online hypergraph coloring (OHC). Firstly, we use a double counting argument to prove an upper bound of the competitive ratio of $FirstFit$ for OVBP. Our proof is conceptually simple, and strengthens the result in \cite{VBP_STOC} by removing the dependency on the bin size parameter.  Secondly, we introduce a notion of an online incidence matrix that is defined for every instance of OHC. Using this notion, we provide a reduction from OHC 
   to OVBP, which allows us to carry known lower bounds of the competitive ratio of algorithms for OHC to OVBP. Our approach significantly simplifies the previous argument from \cite{VBP_STOC} that relied on using  intricate graph structures. In addition, we slightly improve their lower bounds. Lastly, we establish a tight bound of the competitive ratio of algorithms for OHC, where input is restricted to be a hypertree, thus resolving a conjecture in \cite{hyperC_1}. The crux of this proof lies in solving a certain combinatorial partition problem about multi-family of subsets, which might be of independent interest.
\end{abstract}

\section{Introduction and main results}  \label{sec:Intro}

The $\{0,1\}$ $d$-dimensional  \emph{vector bin packing} problem (VBP) with bin size $B$ is defined as follows.

{\noindent\bf Input:} $\cI = \{v_1, \ldots, v_n\}$ where $v_i \in \{0,1\}^d$ for $i=1, \ldots, n$. A bin size parameter $B \in \bbZ^+$.

{\noindent\bf Output:} a partition of $\cI$ into feasible bins.

A subset $X \subseteq \cI$ is called a \emph{feasible bin} if 
    $\sum_{v \in X} v \le \Vec{B}$.
Here and throughout the paper we use the notation $\Vec{B}$ to denote the all-$B$ vector 
    $(B,\ldots,B) \in \bbR^d$.
Let $\OPT(\cI)$ denote the minimum number of bins needed in a feasible bin packing for $\cI$.
If the vector space is $[0,1]^d$ instead of $\{0,1\}^d$, we call the corresponding problem $[0,1]$  VBP.
In the \emph{online} vector bin packing (OVBP), the vectors of $\cI$ arrive online one by one following a given order. An online algorithm $\ALG$ needs to, upon the arrival of every vector $v$, assign  $v$ to a bin,  while maintaining the feasibility. 
Let $\ALG(\cI)$ denote the number of bins used by $\ALG$ on the online instance $\cI$. The competitive ratio of $\ALG$ on $\cI$ is 
    $\ALG(\cI)/\OPT(\cI)$.

Bin packing is a fundamental problem in optimization, for which VBP is one of its many variants. The VBP has applications in virtual machine placement in cloud computing.
While upper bounds on competitive ratios of OVBP algorithms are easier to establish, it took two decades to establish strong lower bounds in the breakthrough \cite{VBP_STOC}. In establishing the lower bounds, \cite{VBP_STOC} built a connection between OVBP and online graph coloring. One part of this work is to show that connection can be generalized to online hypergraph coloring (OHC).
For simplicity and uniformity, we describe vertex coloring in $k$-uniform hypergraph language for $k\ge 2$, while understanding that the case of graphs corresponds to $k=2$.

{\noindent\bf Input:} $H = (V,E)$ is a $k$-uniform hypergraph.

{\noindent\bf Output:} a partition of $V$ into disjoint subsets so that no hyperedge lies entirely in one subset.

Such a partition is called a proper coloring where no hyperedge is monochromatic. Let $\chi(H)$ denote the chromatic number of $H$, i.e., the least number of colors needed in a proper coloring. 
In online hypergraph coloring in the vertex arrival model, the vertices of $H$ arrive online one by one following a given order. Upon the arrival of a vertex, its hyperedges incident to existing vertices are revealed. Note that an online hyperedge is revealed only when its last vertex arrives. An online coloring algorithm $\ALG$ must assign an irrevocable color upon the arrival of each online vertex while maintaining a proper coloring. Let $\ALG(H)$ denote the number of colors used by $\ALG$. The competitive ratio of $\ALG$ on $H$ is 
    $\ALG(H) / \chi(H)$. 
A simple online coloring algorithm is $FirstFit$: color every online vertex with the smallest natural number that does not create any monochromatic hyperedges.
See Figure \ref{fig:online_graph} for an illustration: 
    $FirstFit(G) = FirstFit(H) = 3$, 
but $\chi(G) = \chi(H) = 2$.
For arbitrary graphs (resp. hypergraphs), almost tight upper and lower bounds on the competitive ratio
\cite{graph_upperbd, graph_lowerbd, graph_vishwanathan1992randomized} (resp. \cite{hyperC_1,hyperC_randomized}) are known for both deterministic and randomized algorithms. Unless stated otherwise, the term ``algorithm'' means a deterministic online algorithm.

\begin{figure}[t]        
    \centering
    \includegraphics[scale=.7]{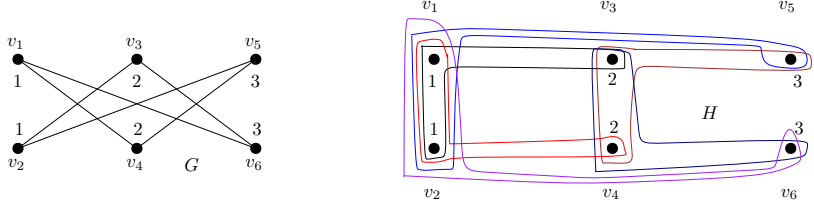}
    \caption{Online graph $G$ and online $3$-uniform hypergraph $H$ colored by $FirstFit$, vertices arriving in order $v_1, \ldots, v_6$.}    \label{fig:online_graph}
\end{figure}


\subsection{Main Results}       \label{sec:Results}

Firstly, we prove some upper bounds on the competitive ratio of $FirstFit$ for OVBP. 
The algorithm $FirstFit$ for $\{0,1\}$ OVBP is the following: upon the arrival of each online vector $v$, put $v$ into the first feasible bin that can accommodate $v$; if none of the previously opened bins can accommodate $v$ then open a new bin and place $v$ into it. A competitive ratio $2B\sqrt{d}$ is given in \cite{VBP_STOC}. In Section \ref{sec:upper_bound}, we use a conceptually simple double counting argument to show an upper bound that is independent of the bin size parameter $B$. Our double counting proof is also applicable to $[0,1]$ OVBP. Furthermore, the double counting method might be useful in analyzing $FirstFit$ for other online problems.

\begin{theorem} \label{thm:double_counting_FF}
    For the $\{0,1\}$ $d$-dimensional OVBP with bin size $B \in \bbZ^+$, the competitive ratio of $FirstFit$ is 
        $< \sqrt{2d} + 2$.
\end{theorem}


Secondly, we give much simpler proofs of improved lower bounds  on the competitive ratio of arbitrary algorithms for OVBP. In \cite{VBP_STOC}, the OVBP lower bounds are proved via a reduction from online graph coloring, instead of online \emph{hypergraph} coloring. To make their reduction work for $B\ge 2$, \cite{VBP_STOC} devise an iterative ``filtering algorithm''
to handle certain $K_{B+1}$-free graphs. In particular, their filtering algorithm relies on exploiting the structure of Halld{\'o}rsson-Szegedy graphs \cite{graph_lowerbd}. It is asked in \cite{VBP_STOC} whether this can be eliminated. We answer this question in the affirmative. We define an online incidence matrix for every online hypergraph, and provide a reduction that directly translates lower bounds of OHC to lower bounds from OVBP, hence significantly simplify the proof of \cite{VBP_STOC}. The idea of the online incidence matrix is implicit in \cite{VBP_STOC}, though this term was not used, and it was only applied to graphs not hypergraphs. As a byproduct of applying lower bounds from OHC directly, we remove the extra $\log$ factor for the case $B\ge 2$ in \cite{VBP_STOC}. 
We later show how these lower bounds can be generalized to $[0,1]$ OVBP and randomized algorithms.

\begin{theorem}     \label{thm:lowerbounds}
    For $\{0,1\}$  $d$-dimensional OVBP with bin size $B\in \bbZ^+$, 
    \begin{enumerate}[(1)]
        \item for $B=1$, the competitive ratio of $FirstFit$ is 
        $> (\sqrt{2d} - 1) / 4$. 

        \item for $B= 1$, the competitive ratio of any online algorithm is 
        $> 2(\sqrt{2d}-1) / (\log_2 (\sqrt{2d} -1))^2$.
        
        \item for $B\ge 2$, the competitive ratio of any online algorithm is 
        $> (((B+1)! \cdot d)^{1/(B+1)} - 1) / 2B$.
    \end{enumerate}
    Furthermore, the lower bounds can be achieved via $d$-dimensional OVBP instances with the following property: for (1) and (3), $\OPT$ of the instance is $2$; for (2), $\OPT$ of the instance is 
        $\le \log_2 (\sqrt{2d}+1)$.
\end{theorem}

Lastly, we establish a tight bound on the competitive ratio of algorithms for coloring online hypertrees. A hypertree is a connected hypergraph without any cycle. A cycle in a hypergraph is a vertex edge sequence
    $v_1 e_1 v_2 e_2 \cdots v_t e_t v_{t+1}$ 
such that 
    $\{v_i, v_{i+1}\} \subseteq e_i$
for every $1 \le i \le t$, and all edges and vertices are distinct except
    $v_1 = v_{t+1}$. 
The authors of 
\cite{hyperC_1} observed that for every online $k$-uniform hypertree $H$ on $n$ vertices, 
    $FirstFit(H) \le 1 + \log_k(n)$.
The proof follows similar steps as the proof for online trees, i.e., $k = 2$, but requires extra arguments.  
It is known that every online coloring algorithms on trees has competitive ratio at least $1+\log_2(n)$, i.e., that of $FirstFit$. This lower bound is established by a forest construction, see, e.g.,
\cite{graph_gyarfas1988line}.  
The authors of \cite{hyperC_1} conjectured that $FirstFit$ also achieves the best possible competitive ratio on online hypertrees. A natural attempt is to generalize the forest construction from trees to hypertrees. The crux lies in solving a certain partition problem about multi-family of subsets. We show in Section \ref{sec:multi-family} that a suitable combination of a forward greedy strategy with a backward greedy strategy solves that partition problem. With this solution, in Section \ref{sec:hypertree} we confirm the conjecture from \cite{hyperC_1}.


\begin{theorem} \label{thm:hypertree_lb}
    Let $\ALG$ be any online coloring algorithm for $k$-uniform hypertrees. Then, there exists a $k$-uniform online hypertree with at most $k^{m-1}$ vertices such that $\ALG$ uses at least $m$ colors.
\end{theorem}




\section{Improved upper bounds for FirstFit of OVBP by double counting}   \label{sec:upper_bound}



\begin{proof}[Proof of Theorem \ref{thm:double_counting_FF}]
    Let $\cI = \{v_1, \ldots, v_n\}$ be an OVBP instance, with the online vectors arriving in that order. 
    Viewing each $v_i$ as a column vector, we identify $\cI$ as a 0-1 matrix $\cI_{d\times n}$. Let $\OPT(\cI)$ and $\FF(\cI)$ denote the numbers of bins in offline OPT and  $FirstFit$, respectively.  Let $w(\cI)$ denote the number of $1$'s in $\cI$. We use double counting to count $w(\cI)$.  Let $p(\cI)$ denote the maximum number of $1$'s in a row of $\cI$. If $p(\cI) = 1$, then $\OPT(\cI) = \FF(\cI) = 1$.  So, we assume $p(\cI) \ge 2$. For  simplicity, we omit $\cI$ in $w(\cI)$ and other notations. Obviously, $\OPT \ge p/B$.
    
    Counting by rows of $\cI$, trivially, we have $w \le d \cdot p$. Next, we count by columns of $\cI$. In fact, we will count by bins used by $FirstFit$. We say a vector $v \in \{0,1\}^d$ is {\bf unfit} for bin $j$ if $v$ cannot be put into bin $j$ by $FirstFit$. By the definition of $FirstFit$, this means that there exists a row $r$ such that bin $j$ is full (i.e., has $B$ $1$'s) in  row $r$, and the entry of $v$ in  row $r$ is $1$.
    Since each row contains at most $p$ many $1$'s, each entry of value $1$ in $v$ can be unfit for $\le (p - 1) / B$ many bins. If $FirstFit$ puts the online vector $v$ into bin $i$, i.e., $v$ is unfit for $i-1$ bins, then, $v$ must contain at least 
        $(i-1) / ((p - 1) / B) = B(i-1) / (p - 1)$
    many $1$'s. 
    Observe also that for 
        $i=1, \ldots, \FF-1$, 
    bin $i$ contains at least $B$ vectors. Therefore, for 
        $1 \le i \le \FF-1$,  
    bin $i$ contains at least 
        $\frac{B^2(i-1)}{p - 1}$ 
    many $1$'s.  
    The last bin, i.e., bin $\FF$, contains at least one vector, which contains at least $B(\FF-1)/(p-1)$ many $1$'s.
    Summing over all bins (i.e., counting by `column') and applying the upper bound on $w$, we get 
    \begin{equation}    \label{eq:FF_ineq}
            dp \ge w
            \ge \left(\sum_{i=1}^{\FF-1} \frac{B^2(i-1)}{p - 1} \right)
            + \frac{B(\FF-1)}{p - 1} 
            = \frac{1}{p - 1} \cdot \left( B^2 \cdot \frac{(\FF-1)(\FF-2)}{2} + B (\FF-1)
            \right).
    \end{equation}
    Solving inequality \eqref{eq:FF_ineq}, we get
        $\FF \le \frac{1}{2B} \cdot 
        \Big( \sqrt{8dp^2 - 8dp + (B-2)^2} + 3B - 2 \Big)$.
    Note that
        $8dp^2 - 8dp < (2\sqrt{2d}p-1)^2$
    holds for every $d,p \ge 1$,  hence,
        $\FF < \frac{(2\sqrt{2d}p - 1) + |B-2| + 3B -2 }{2B}$.
    For $B=1$, this gives
        $\FF 
        < \sqrt{2d} p + 0.5
        \le \sqrt{2d} \cdot \OPT + 0.5$.
    For $B\ge 2$, this gives
        $\FF 
        < \sqrt{2d} \cdot p/B + 2(B-1.25) /B
        \le \sqrt{2d} \cdot \OPT + 2(B-1.25) /B$.
\end{proof}

By the same double counting proof, one can prove a similar bound for the $[0,1]$ OVBP. 

\begin{theorem} \label{thm:double_counting_FF_real}
    For the $[0,1]$ $d$-dimensional OVBP with bin size $B \in \bbR^+$ and $B > 1$, the competitive ratio of $FirstFit$ is 
        $< \sqrt{2d} \cdot \frac{B}{B-1} + \frac{2}{\OPT} \le 2\sqrt{2d} + 2$.
\end{theorem}

\section{Lower bounds for OVBP via the online incidence matrix of an online hypergraph}      \label{sec:lower_bounds}


Let $H$ be an online $k$-uniform hypergraph of $n$ vertices. We associate $H$ with an \emph{online incidence matrix} denoted by $I_H$ of dimension ${n\choose k} \times n$, as follows. Label the rows by all ${n\choose k}$ hyperedges, label the columns by vertices according to their order of arrival in $H$. We use $v \in e$ to mean that $v$ is a vertex of hyperedge $e$. 
Upon the arrival of an \emph{online} vertex $v$, define the \emph{online} column vector for $v$ as follows.
\begin{equation}    \label{eq:def_online_incidence_matrix}
    I_H(e,v) 
    = 
    \begin{cases}
        0,  \quad &\text{if } v \not\in e, \\
        1, \quad &\text{if } v \in e, \text{ but not all vertices of $e$ have arrived}, \\
        0,  \quad &\text{if } v \in e, \text{ and $v$ is the last vertex of $e$, but hyperedge $e \not\in E(H)$}, \\
        1,  \quad &\text{if } v \in e, \text{ and $v$ is the last vertex of $e$, and hyperedge $e \in E(H)$}.
    \end{cases}
\end{equation}
Let $e(I_H)$ denote the sum of entries in row $e$, i.e., the number of $1$'s. Obviously,
\begin{equation}    \label{eq:only_two_values}
    e(I_H) 
    =
    \begin{cases}
        k-1, \quad & e\not\in E(H), \\
        k, \quad  & e\in E(H).
    \end{cases}
\end{equation}

For example, let $K$ be the online subgraph of $G$ in Figure \ref{fig:online_graph}, on only the first four vertices. 
Then, $I_K$ is given in \eqref{eg:example_online_incidence_matrix}, where the offline incidence matrix $M_K$ is also given for the purpose of comparison. For the online hypergraph $H$ in Figure \ref{fig:online_graph}, the online column vector corresponding to $v_4$ is given in \eqref{eg:example_online_incidence_matrix}, where $v_4^T$ denotes the transpose of $v_4$.
\begin{equation}    \label{eg:example_online_incidence_matrix}
\begin{split}
    M_K &= 
    \begin{pNiceMatrix}[first-row,first-col]
               & v_1 & v_2 & v_3 & v_4 \\
    v_1 v_4 & 1 & 0 & 0 & 1 \\
    v_2 v_3 & 0 & 1 & 1 & 0 
    \end{pNiceMatrix},  
    \qquad
    I_K
    =
    \begin{pNiceMatrix}[first-row,first-col]
               & v_1 & v_2 & v_3 & v_4 \\
    v_1 v_2 & 1 & 0 & 0 & 0 \\
    v_1 v_3 & 1 & 0 & 0 & 0 \\
    v_1 v_4 & 1 & 0 & 0 & 1 \\
    v_2 v_3 & 0 & 1 & 1 & 0 \\
    v_2 v_4 & 0 & 1 & 0 & 0 \\
    v_3 v_4 & 0 & 0 & 1 & 0 
    \end{pNiceMatrix}.   \\
    v_4^T &= 
    \begin{pNiceMatrix}[first-row]
                 v_1 v_2 v_4 
               & v_1v_3v_4
               & v_1v_5v_4
               & v_1v_6v_4
               & v_2v_3v_4
               & v_2v_5v_4
               & v_2v_6v_4
               & v_3v_5v_4
               & v_3v_6v_4
               & v_5v_6v_4 
               & \text{other}
               & \cdots \\
                 1
               & 0
               & 1
               & 1
               & 0
               & 1
               & 1
               & 1
               & 1
               & 1
               & 0
               & \cdots 
    \end{pNiceMatrix}. 
\end{split}
\end{equation}


Hence, given an online uniform hypergraph $H$ in the vertex arrival model, the online incidence matrix $I_H$ can be constructed \emph{online} column by column, where each online column vector of $I_H$ corresponds to an online vertex  of $H$. Conversely, given the online incidence matrix $I_H$ of some online uniform hypergraph $H$, the online uniform hypergraph  $H$ can also be recovered \emph{online} vertex by vertex.


By presenting column vectors of $I_H$ column by column, the online incidence matrix $I_H$ is an instance of $\{0,1\}$ OVBP.
Let $H$ be a  $k$-uniform online  hypergraph on $n$ vertices, 
    $k\ge 2$. 
Let 
    $d = {n \choose k}$,
    $B = k-1 \ge 1$.
Let $\ALG$ be an algorithm solving the $\{0,1\}$ $d$-dimensional OVBP with bin size $B$. For an online vector 
    $v \in \{0,1\}^d$, 
let 
    $\ALG(v)$
denote the bin to which $v$ is assigned. 
We define an online algorithm $\widetilde{\ALG}$ for coloring the $k$-uniform online hypergraph $H$, as follows. By an abuse of notation, we use $v$ to denote both an online vertex $v\in V(H)$ and its corresponding online vector $v \in \{0,1\}^d$. 
Upon the arrival of online vertex $v \in V(H)$, construct the online vector $v \in \{0,1\}^d$ defined by \eqref{eq:def_online_incidence_matrix}, run $\ALG$ on $v$, and set 
    $\widetilde{\ALG}(v) = \ALG(v)$.
    
\begin{lemma} \label{lem:reduction_correct}
    $\widetilde{\ALG}$ gives a proper coloring of the online hypergraph $H$. Furthermore,
        $\widetilde{\ALG}(H) = \ALG(I_H)$
    and 
        $\chi(H) = \OPT(I_H)$.
\end{lemma}

\begin{proof}
    By definition of $\widetilde{\ALG}$, 
        $\widetilde{\ALG}(H) = \ALG(I_H)$. 
    We show that a feasible bin packing of $I_H$ corresponds to a proper coloring of $H$, and vice versa. Indeed, by \eqref{eq:only_two_values}, bin of size $B=k-1$ can never contain an entire edge of $H$. Hence, a feasible bin packing of $I_H$ gives a proper coloring of $H$, implying 
        $\chi(H) \le \OPT(I_H)$. 
    On the other hand, let 
        $S \subseteq V(H)$ 
    be a color class in a proper coloring of $H$. So, $S$ does not contain any entire edge $e \in E(H)$. Again, by \eqref{eq:only_two_values}, 
        $\sum_{v \in S} v \le \Vec{B}$.
    Hence, a proper coloring of $H$ gives a feasible bin packing of $I_H$, implying 
        $\OPT(I_H) \le \chi(H)$.
\end{proof}

Note that $\widetilde{\ALG}$ needs to know $n$, the number of vertices, in advance, whereas OHC was defined with $n$ being unknown to the algorithm. While it may be an interesting question of how much knowledge of $n$ can affect competitive ratio, it does not affect our results established below, since the lower bounds on OHC that we rely on work in the regime of $n$ known by the algorithm in advance. We are now ready to prove Theorem \ref{thm:lowerbounds}. 

        

\begin{proof}[Proof of Theorem \ref{thm:lowerbounds}]
    Let $k=B+1$. 
    Choose $n$ to be the largest integer such that
        ${n \choose k} \le d$,
    then, 
        $d < {n+1 \choose k} \le \frac{(n+1)^k}{k!}$,
    i.e.,
        $n > (k! \cdot d)^{1/k} - 1$. 
    Let
        $d' = {n \choose k}$.
    To show a lower bound for online $d'$-dimensional VBP, by Lemma \ref{lem:reduction_correct}, it suffices to provide a corresponding online coloring lower bound.
    To handle $d \ge d'$, simply add extra $d-d'$ many $0$'s into $d'$-dimensional vectors to get a $d$-dimensional VBP instance. Clearly, the same lower bound holds.  
        
    (1) $k=B+1=2$. When $\ALG$ is $FirstFit$ for the  OVBP, by our definition,  $\widetilde{ALG}$ is exactly  $FirstFit$ for online graph coloring. Generalizing $G$ in Figure \ref{fig:online_graph}, we get an online graph on $n$ vertices that is bipartite, but $FirstFit$ uses at least $\ceil{n/2} \ge n/2$ colors. 
    
    (2) $k=B+1=2$. In \cite{graph_lowerbd}, it is shown that for every online graph coloring algorithm, there exists an online graph on $n$ vertices that is $\log_2n$-colorable, but the algorithm uses at least 
        $2n / \log_2 n$
    colors. 
        
    (3) $k=B+1 \ge 3$. In \cite{hyperC_1}, it is shown that for every online hypergraph coloring algorithm, there exists an online $k$-uniform  hypergraph on $n$ vertices that is $2$-colorable, but the algorithm
    uses 
        $\ceil{n/(k-1)}$
    colors.
\end{proof}

The $[0,1]$ case can be handled by adapting a trick introduced in \cite{VBP_STOC}. Specifically, we modify the definition of \eqref{eq:def_online_incidence_matrix}  to get an online incidence matrix $I'_H$ of dimension ${n\choose k-1} \times n$. The rows of $I'_H$ will be labeled by $\epsilon \in {[n] \choose k-1}$, instead of hyperedges of $H$. By ``$\epsilon$ all arrived'', we mean that all $k-1$ vertices in $\epsilon$ have arrived. 
\begin{equation}    \label{eq:def_online_incidence_matrix_realcase}
    I'_H(\epsilon,v) 
    = 
    \begin{cases}
        0,  \quad &\text{if $v \not\in \epsilon$, and $\epsilon$ not all arrived}, \\
        1, \quad &\text{if $v \in \epsilon$}, \\
        0,  \quad &\text{if $v \not\in \epsilon$,  and $\epsilon$ all arrived, and $\epsilon \cup \{v\} \not\in E(H)$}, \\
        1/n,  \quad &\text{if $v \not\in \epsilon$,  and $\epsilon$ all arrived, and $\epsilon \cup \{v\} \in E(H)$}.
    \end{cases}
\end{equation}
By feeding $I'_H$ to the $[0,1]$ OVBP, one can similarly prove Lemma \ref{lem:reduction_correct}, hence the lower bounds follow. Indeed, let $\epsilon(I'_H)$ denote the sum of entries in row $\epsilon$, it is easy to see that 
    $k-1 \le \epsilon(I'_H) \le k-1 + (n-k+1)/n < k$.
Furthermore, if 
    $e = \epsilon \cup \{v\}$ 
is an online hyperedge of $H$ in which $v$ arrives the latest, then 
    $\sum_{u \in e} I'_H(\epsilon, u) = k-1+ 1/n > k-1$. 
Setting $B=k-1$ would guarantee that a feasible bin cannot contain any entire hyperedge $e \in E(H)$, hence feasible bins correspond to color classes, and vice versa. 

\begin{theorem}     \label{thm:lowerbounds_realcase}
    For $[0,1]$  $d$-dimensional OVBP with bin size $B\in \bbZ^+$, 
    \begin{enumerate}[(1)]
        \item for $B=1$, the competitive ratio of $FirstFit$ is 
        $\ge d / 4$. 

        \item for $B= 1$, the competitive ratio of any online algorithm is 
        $> 2d / (\log_2 d)^2$.

        \item for $B\ge 2$, the competitive ratio of any online algorithm is 
        $> ((B! \cdot d)^{1/B} - 1) / 2B$.
    \end{enumerate}
    Furthermore, the lower bounds can be achieved via $d$-dimensional OVBP instances with the following property: for (1) and (3), $\OPT$ of the instance is $2$; for (2), $\OPT$ of the instance is 
        $\le \log_2 d$.
\end{theorem}

\begin{proof}
    Let $k =B+1$. Now, we choose $n$ to be the largest integer such that
        ${n \choose k-1} \le d$,
    then, 
        $d < {n+1 \choose k-1}$,
    i.e.,
        $n > ((k-1)! \cdot d)^{1/(k-1)} - 1
        = (B! \cdot d)^{1/B} - 1$.
    The rest is the same.
\end{proof}

Finally, it is not hard to see that lower bounds for randomized online coloring algorithms also directly translate to lower bounds for OVBP randomized algorithms, via our online incidence matrix reduction. 

\begin{theorem}     \label{thm:lowerbounds_randomized_alg}
    For $\{0,1\}$ (resp. $[0,1]$) $d$-dimensional OVBP with bin size $B\in \bbZ^+$, 
    \begin{enumerate}[(1)]
        \item for $B= 1$, the competitive ratio of any randomized online algorithm is 
            $> (\sqrt{2d}-1) / 16 (\log_2 (\sqrt{2d} -1))^2$ 
        (resp. 
            $> d / 16 (\log_2 d)^2$).
        
        \item for $B\ge 2$, the competitive ratio of any randomized online algorithm is 
            $> (((B+1)! \cdot d)^{1/(B+1)} - 1) / 4(B+1)$
        (resp. 
            $> ((B! \cdot d)^{1/B} - 1) / 4(B+1)$).
    \end{enumerate}
    
    Furthermore, the lower bounds can be achieved via $d$-dimensional OVBP instances with the following property: for (1), $\OPT$ of the instance is at most 
        $\log_2 (\sqrt{2d} -1)$
        (resp. $4\log_2 d$); 
    for (2), $\OPT$ of the instance is 
        $2$ (resp. 2).
\end{theorem}

\begin{proof}
    The lower bound for randomized online graph coloring algorithm comes from \cite{graph_lowerbd}. 
    The lower bound for randomized online hypergraph coloring algorithm comes from \cite{hyperC_randomized}. 
\end{proof}

We remark that Theorem \ref{thm:lowerbounds_realcase}-(2) and Theorem \ref{thm:lowerbounds_randomized_alg}-(1) haven been proved in \cite{VBP_STOC}, while Theorem \ref{thm:lowerbounds_realcase}-(3) and Theorem \ref{thm:lowerbounds_randomized_alg}-(2) improve the corresponding bounds in \cite{VBP_STOC}.



\section{A multi-family partition problem}  \label{sec:multi-family}

Let $C \subseteq \bbN$. Let $\cF$ be a family of subsets of $C $. We allow multiple identical subsets in $\cF$, i.e., $\cF$ itself is a multiset. We call $\cF$ a \emph{multi-family} on $C$.
Let $|\cF|$ denote the cardinality of $\cF$. For example, the following $\cF$ is a multi-family on $\{1,2,3\}$ of cardinality $|\cF| = 12$: 
    \begin{equation}    \label{eq:example_multifamily}
        \cF = \big\{\{1\}, \{2\}, \{3\}, \{1\}, \{1,3\}, \{1,3\}, \{2,3\}, \{1,3\}, \{1,2, 3\}, \{1,2, 3\}, \{1,2, 3\}, \{1,2, 3\} \big\}.
    \end{equation}

\begin{definition}  \label{def:multi-family-diverse-star}
    Let 
        $C = \{c_1, \ldots, c_q\} \subseteq \bbN$.
    Let $\cF$ be a multi-family on $C$. 
    Let $p \in \bbN$. 
    We say $\cF$ is \emph{$(p,C)$-diverse} if there exists a partition 
        $\cF = \cF_1 \cup \cdots \cup \cF_q$
    such that for every $1 \le i \le q$: 
        (i) $|\cF_i| \ge p$, and
        (ii) $|T| \ge i$ for every $T \in \cF_i$. 
    We say $\cF$ is \emph{$(p, C)$-starry} if instead of condition (ii) we require (ii') $c_i \in T$ for every $T \in \cF_i$.
\end{definition}


Sometimes we omit the notation $(p,C)$ when it is clear from the context. For convenience, we call each part in the diverse or starry partition a \emph{block}. In the starry partition, we call block $\cF_i$ a $c_i$-star. Clearly, $\cF$ in \eqref{eq:example_multifamily}  admits a $(4,\{1,2,3\})$-diverse partition. In fact, $\cF$ also admits a $(4,\{1,2,3\})$-starry partition: 
\begin{equation}    \label{eq:starry-partition}
    \cF = \big\{\{1\}, \{1\}, \{1,3\}, \{1,3\} \big\} 
    \cup
    \big\{\{2\}, \{2,3\}, \{1,2, 3\}, \{1,2, 3\} \big\}
    \cup
    \big\{\{3\}, \{1,3\}, \{1,2, 3\}, \{1,2, 3\} \big\}.
\end{equation}

A $(p,C)$-starry multi-family needs not be $(p,C)$-diverse. For example, 
     $\cF = \big\{ \{1\}, \{1\}, \{2\}, \{2\} \big\}$ 
is $(2,\{1,2\})$-starry but is not $(2,\{1,2\})$-diverse.

\begin{theorem} \label{thm:diverse-implies-starry}
    For every $p\ge 1$, for every $C \subseteq \bbN$, every $(p,C)$-diverse multi-family is also $(p,C)$-starry.
\end{theorem}

\begin{proof}
    We use induction on $q = |C|$. The base case $q=1$ is trivial. Assume the claim is true for every $C$ of size $q$. Now, consider
        $C = \{c_1, \ldots, c_{q+1}\}$.
    Let $\cF$ be a $(p, C)$-diverse multi-family on $C$ with a diverse partition
        $\cF = \cF_1 \cup \cdots \cup \cF_{q+1}$.
    From this diverse partition, we use a greedy strategy to construct a starry partition, as follows. 
        
    Respecting the order 
        $\cF_1, \cdots, \cF_{q+1}$,
    successively and greedily choose subsets containing $c_1$ from each $\cF_i$ as many as possible, until we reach a sub-multi-family $\cG$ of cardinality $|\cG| = p$. Note that it is always possible to find such $\cG$ of cardinality $|\cG| = p$: this is because $\cF$ is diverse, in particular, $\cF_{q+1}$ is a sub-family of $\cF$ such that every $T \in \cF_{q+1}$ contains $c_1$ and $|\cF_{q+1}| \ge p$.         
    Now, $\cG$ is a desired $c_1$-star, to construct the other stars, we consider $\cF - \cG$.
    Define a map 
        $\phi: 2^{C} \to 2^{C'}$
    by 
        $\phi(T) = T-\{c_1\}$.
    Note that 
        $\phi(C) = C - \{c_1\}$,
    hence, $|\phi(C)| = q$. 
    We use the notation $\phi(\cF)$ to denote the resulting multi-family on $\phi(C)$. 
    Observe that it suffices to show 
        the multi-family $\phi(\cF - \cG)$ is $(p,\phi(C))$-diverse. 
    By induction, this implies 
        $\phi(\cF - \cG)$ is $(p,\phi(C))$-starry.
    Clearly, by pulling back the starry partition for $\phi(\cF - \cG)$ together with $\cG$ gives a $(p,C)$-starry partition for $\cF$.

    It remains to show $\phi(\cF - \cG)$ is $(p,\phi(C))$-diverse. Recall the diverse partition
        $\cF = \cF_1 \cup \cdots \cup \cF_{q+1}$.
    For each $\cF_i$, write
        $\cF_i = \cF_i^g \cup \cF_i^r$,
    where
        $\cF_i^g = \cF_i \cap \cG$
    and
        $\cF_i^r = \cF_i - \cF_i^g$.
    By the choice of $\cG$,      
        $\sum_{i=1}^{q+1} |\cF_i^g| = |\cG| = p\ge 1$.
    Let $1 \le j^*\le q+1$ be the largest index for which $\cF_{j^*}^g > 0$, in other words, $\cF_{j^*}$ is the place at which the greedy choice of $\cG$ stops.
    We have
    \begin{equation}    \label{eq:reduced-family}
        \phi(\cF - \cG) = \phi(\cF_1^r) \cup \cdots \cup \phi(\cF_{j^*}^r) \cup \cdots \cup \phi(\cF_{q+1}^r).
    \end{equation}        

    {\noindent Case 1:} $j^* + 1 \le i \le q+1$. For such $i$, by the choice of $j^*$, we have
           $\cF_i^g = \emptyset$, 
    or equivalently,
        $\cF_i^r = \cF_i$.
    Hence,
        $|\phi(\cF_i^r)| = |\phi(\cF_i)| = |\cF_i| \ge p$.
    Since every $T \in \cF_i$ satisfies
        $|T| \ge i$,
    we have every $\phi(T) \in \phi(\cF_i)$ satisfies
        $|\phi(T)| \ge i-1$.
    These 
        $\phi(\cF_{j^* +1}^r), \ldots, \phi(\cF_{q+1}^r)$ 
    will be the last $q - (j^* -1)$ blocks for a diverse partition of $\phi(\cF - \cG)$

    {\noindent Case 2:} $1 \le i \le j^*$. Consider those $i< j^*$ first. Since we choose $\cG$ greedily, this means
        $c_1 \not\in T$    for every $T \in \cF_i^r$, 
    since otherwise $T$ would have been chosen into $\cF_i^g$. This implies 
        $\phi(T) = T$ for every $T \in \cF_i^r$.
    Hence,
       $|\phi(T)| = |T| \ge i$ for every $T \in \cF_i^r$.
    Let us also note that, simply by definition, 
        $|\phi(T)| \ge j^* -1$ for every $T \in \cF_{j^*}^r$.            
    Now, consider the following $j^* -1$ objects:
        $\phi(\cF_1^r), \ldots, \phi(\cF_{j^*-1}^r) \cup \phi(\cF_{j^*}^r)$.
    What we have shown is that every $\phi(T)$ in the $i$-th object in the above list satisfies 
        $|\phi(T)| \ge i$. 
    However, these $j^* -1$ objects may not satisfy Condition (i) for the diverse partition. For this purpose, we use a \emph{backward} greedy strategy to reallocate: starting from the last (i.e., the $(j^*-1)$-th) object, keep arbitrarily $p$ subsets, and move the rest to the next, and repeat. Note that after this reallocation, every subset in the $i$-th block still has size at least $i$.  It only remains to show that this backward greedy strategy is feasible, in other words, at every step $1 \le k \le j^*-1$, the current object contains at least  $p$ subsets. Indeed, consider 
        $\phi(\cF_{k}^r), \ldots, \phi(\cF_{j^*-1}^r) \cup \phi(\cF_{j^*}^r)$.
    The total cardinality (i.e., the total number of subsets in these objects) is
    \[
        \sum_{i=k}^{j^*} |\phi(\cF_{i}^r)|
        = \sum_{i=k}^{j^*} |\cF_{i}^r|
        \ge \sum_{i=k}^{j^*} (p - |\cF_{i}^g|)
        \ge (j^* - k)p,
    \]
    where we used the fact that
        $\sum_{i=k}^{j^*} |\cF_{i}^g| \le |\cG| = p$.
    Since before step $k$, we have chosen exactly 
        $(j^* - 1 - k) p$ 
    subsets, therefore, at least $p$ subsets are left to be chosen at step $k$, as desired. 

    The blocks in Case 1 and Case 2 together give a $(p,\phi(\Omega))$-diverse partition for $\phi(\cF-\cG)$, as desired.
\end{proof}

\section{A lower bound for online coloring uniform hypertrees}  \label{sec:hypertree}

With Theorem \ref{thm:diverse-implies-starry}, we are ready to prove Theorem \ref{thm:hypertree_lb}. Our proof strategy is illustrated in Figure \ref{fig:hypertree}.
    \begin{figure}[t]        
        \centering
        \includegraphics[scale=.6]{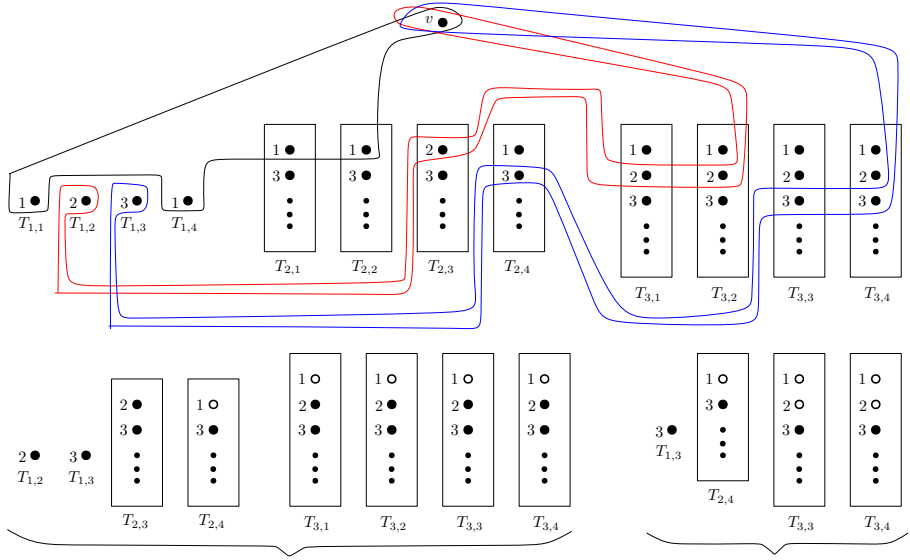}
        \caption{A vertex $v$ merges $F_3$ by a greedy selection of hyperedges. Here, $k=5, m=4$. The multi-family in \eqref{eq:example_multifamily} is the multi-family of color subsets of $F_3$, where its starry partition in \eqref{eq:starry-partition} corresponds to the hyperedge selection here. Below, the two multi-families of color subsets demonstrate how the inductive proof is carried out in the proof of Theorem \ref{thm:diverse-implies-starry}.}    \label{fig:hypertree}
    \end{figure}
    
    
\begin{proof}[Proof of Theorem \ref{thm:hypertree_lb}]
    Let $T_m$ denote the online hypertree of size $\le k^{m-1}$ that we will construct to satisfy 
        $\ALG(T_m) \ge m$.    
    We use induction on $m\ge 1$. 
    For base case $m=1$, $T_1$ consists of a single vertex, and this clearly works. 
    For $m \ge 2$, we construct $T_m$ inductively in two steps: 
    
    {\noindent STEP 1:} for every 
        $i=1, \ldots, m-1$, 
    successively construct $k-1$ disjoint copies of $T_i$, let $F_{m-1}$ denote this hyperforest consisting of exactly 
        $(m-1)(k-1)$ 
    disjoint hypertrees. $F_{m-1}$ is naturally partitioned into $m-1$ blocks $B_1, \ldots, B_{m-1}$, where $B_i$ consists of $k-1$ disjoint copies of $T_i$.
    
    {\noindent STEP 2:} create a new vertex $v$ to merge the disjoint hypertrees in $F_{m-1}$ into a single hypertree $T_m$. 

    By induction hypothesis, each copy of $T_i$ has size at most $k^{i-1}$ and $\ALG$ uses at least $i$ colors on it. Hence, 
    the number of vertices of $T_m$ is at most
        $1+\sum_{i=1}^{m-1} (k-1) k^{i-1} = k^{m-1}$
    as desired. 
    Also, 
        $\ALG(F_{m-1}) \ge \ALG(T_{m-1}) \ge m-1$.
    
    {\noindent Case 1:} $\ALG(F_{m-1}) \ge m$. In this case, STEP 2 can be done in an arbitrary way, e.g., for each block $B_i$, for every $T \in B_i$, arbitrarily choose a vertex from $T$, let $e_i$ be the edge containing $v$ together with these $k-1$ vertices from block $B_i$. Obviously, the obtained hypergraph $T_m$ is a hypertree satisfying
        $\ALG(T_m) \ge \ALG(F_{m-1}) \ge m$.
    
    {\noindent Case 2:} $\ALG(F_{m-1})=m-1$. Without loss of generality, assume
        $\ALG(F_{m-1}) = [m-1]$.
    For every $1\le j \le k-1$, 
    Let $T_{i,j}$ denote the set of colors used by $\ALG$ on the $j$-th copy of $T_i$. Then, 
        $T_{i,j} \subseteq [m-1]$,
    and by induction 
        $|T_{i,j}| \ge i$. 
    Let $\cT$ be the multi-family consisting of all these color subsets $T_{i,j}$, for
        $i=1,\ldots, m-1$ and $j=1, \ldots ,k-1$. 
    Obviously, the block partition 
        $F_{m-1} = B_1 \cup \cdots \cup B_{m-1}$
    induces a $(k-1, [m-1])$-diverse partition for the multi-family $\cT$.
    By Theorem \ref{thm:diverse-implies-starry}, 
        $\cT$ is also $(k-1, [m-1])$-starry.
    Consider a starry partition: 
        $\cT = \cT_1 \cup \cdots \cup \cT_{m-1}$.
    Since every color subset in $\cT_i$ contains color $i$, equivalently, this means that the corresponding hypertree contains a vertex with color $i$. 
    Construct a $k$-uniform hyperedge $e_i$ incident to $v$, by picking $k-1$ vertices of color $i$, one vertex from each of these hypertrees corresponding to the block $\cT_i$. Clearly, these $m-1$ hyperedges $e_1, \ldots, e_{m-1}$ force $\ALG$ to use a color that is not in $[m-1]$. Hence, $\ALG(T_m) \ge m$ as desired.
\end{proof}   

\section*{Acknowledgements}


We thank the anonymous referees of LAGOS 2023, whose comments led to a much improved version of the paper.

\bibliography{mybib}{}

\begin{thebibliography}{1}

\bibitem{VBP_STOC}
Yossi Azar, Ilan~Reuven Cohen, Seny Kamara, and Bruce Shepherd.
\newblock Tight bounds for online vector bin packing.
\newblock In {\em Proceedings of the forty-fifth annual ACM symposium on Theory
  of Computing}, pages 961--970, 2013.

\bibitem{graph_gyarfas1988line}
Andr{\'a}s Gy{\'a}rf{\'a}s and Jen{\"o} Lehel.
\newblock On-line and first fit colorings of graphs.
\newblock {\em Journal of Graph theory}, 12(2):217--227, 1988.

\bibitem{hyperC_randomized}
Magn{\'u}s~M Halld{\'o}rsson.
\newblock Online coloring of hypergraphs.
\newblock {\em Information processing letters}, 110(10):370--372, 2010.

\bibitem{graph_lowerbd}
Magn{\'u}s~M Halld{\'o}rsson and Mario Szegedy.
\newblock Lower bounds for on-line graph coloring.
\newblock {\em Theoretical Computer Science}, 130(1):163--174, 1994.

\bibitem{graph_upperbd}
L{\'a}szl{\'o} Lov{\'a}sz, Michael Saks, and William~T Trotter.
\newblock An on-line graph coloring algorithm with sublinear performance ratio.
\newblock {\em Discrete Mathematics}, 75(1-3):319--325, 1989.

\bibitem{hyperC_1}
Judit Nagy-Gy{\"o}rgy and Cs~Imreh.
\newblock Online hypergraph coloring.
\newblock {\em Information Processing Letters}, 109(1):23--26, 2008.

\bibitem{graph_vishwanathan1992randomized}
Sundar Vishwanathan.
\newblock Randomized online graph coloring.
\newblock {\em Journal of algorithms}, 13(4):657--669, 1992.

\end{thebibliography}
\bibliographystyle{plain}

\end{document}